\documentclass{cc}

\usepackage[ansinew]{inputenc}
\usepackage[french, english]{babel}
\usepackage{latexsym,amssymb,amsmath,amsfonts,amscd,color,epsfig,amsthm}

\usepackage[bold,full]{autrecomplexity}
\usepackage{pst-all}

\def\ket#1{{\lvert}#1\rangle}

\def\tCOL{{\bf{3COL}}}
\def\tSAT{{\bf{3SAT}}}

\newcommand{\braket}[2]{\langle#1|#2\rangle}

\newcommand{\pare}[1]{\left(#1\right)}

\newcommand{\demi}{\frac{1}{2}}

\newcommand{\tensor}{\otimes}
\newcommand{\prob}[1]{\left|#1\right|^{2}}

\contact{blierhug@iro.umontreal.ca}% Insert contact email address here.
\received{August 2009 }% Insert received date here.
\title{A Quantum Characterization\\ of $\NP$}% Insert title here. (Use \\ to split lines.)
\author{Hugue Blier\\ Université de Montréal\\ blierhug@iro.umontreal.ca \and
		Alain Tapp\\Université de Montréal\\ tappa@iro.umontreal.ca}% Insert author list here. (Each author must be given
\begin{abstract}
In this article we introduce a new complexity class called \linebreak $\BT$. 
Informally, this is the class of languages for which membership has a
logarithmic-size quantum proof with perfect completeness and soundness which is polynomially close to 1 in a context where the verifier is provided a proof with two unentangled parts. 
We then show that $\BT=\NP$. 
For this to be possible, it is important, when defining the class, 
not to give too much power to the verifier. This result, when
compared to the fact that $\QML=\BQP$, 
gives us new insight on the power of quantum information and the impact of entanglement.
\end{abstract}


\begin{subject}
 68Q15    	Complexity classes
\end{subject}

\begin{document}

\section{Introduction}

In classical complexity, the concept of proof is extensively used to define very interesting complexity classes such as $\NP$, $\MA$ and $\IP$.  When
allowing the verifier  (and the prover) to be quantum mechanical, we obtain complexity classes such as $\QMA$ and $\QIP$.  Quantum complexity classes can
sometimes turn out to have surprising properties.  For example, in contrast with the classical case, we know that quantum interactive proofs can be
restricted to three messages; that is,  $\QIP=\QIP${\bf (3)} \citep{KW}.

Because of the probabilistic nature of quantum computation, the most natural quantum generalization of $\NP$ is $\QMA$.  This is the class of
languages having polynomial size quantum proofs.  
A quantum proof obviously requires a quantum verifier, but behaves similarly to a classical proof
with regards
to completeness and soundness.  Since group non-membership is in $\QMA$ \citep{053}, but is not known to be in $\MA$ (and therefore $\NP$), we have an
example of a statement having polynomial-size quantum proofs but no known polynomial-size classical proof.

In this paper, we are interested in logarithmic-size quantum proofs. Classically, when considering a
polynomial-time verifier, the concept of logarithmic-size classical proofs is not interesting.  
Any language having logarithmic-size classical proofs
would also be in $\P$, since one can go through {\em every} possible logarithmic-size proof in polynomial time. 

In the quantum case, very short quantum proofs could still be interesting.  Any reasonable classical description of a quantum proof requires a
polynomial number of bits and thus one cannot try all quantum proofs using a classical simulator.  
That being said, if the verifier is simple enough, 
the optimization problem of finding a proof that makes the verifier accept with high enough probability can be turned into a semidefinite
programming problem \citep{043,044} of polynomial size. Thus, if the verifier is simple enough, then the language is in $\P$.  
Also, if the verifier is in $\BQP$, then one still only obtains $\BQP$ \citep{031}.

Although we just argued that logarithmic-size classical and quantum proofs seem uninteresting, by slightly changing the rules of the game, we get an
interesting complexity class.  In preliminary work \citep{cancun}, we showed that $\NP \subseteq \QMLd$. This class is also defined with the promise
that two logarithmic-size unentangled
registers are given to the verifier. This promise gives the verifier more leeway to check the proof and limits the prover's ability to cheat.
Therefore, this gives a new perspective on the properties
of entanglement. 

In parallel with our work, and independently of us,  Aaronson et al.~\citep{Aaronson} have shown that $\tSAT$ is in $\QMLr$ (i.e. with
$\sqrt{n}\polylog(n)$ unentangled registers) with constant completeness
and soundness.  It seems that only two unentangled registers for the certificate are not enough
to check the proof with a constant gap between completeness and soundness; we achieve perfect completeness but soundness polynomially close to one.
They commented on our previous
note \citep{note} emphasizing that this soundness cannot be improved to a constant unless $\QMA(2) = \NEXP$. Note also that since the length of the
proof is logarithmic and the number of registers is constant, showing that $\tCOL$ (the language of graphs colorable with three colors) is in $\QMLd$ would implies that $\NP \subseteq \QMLd$. Therefore, constant
soundness is achieved in \citep{Aaronson} at the cost that their result cannot be generalized to $\NP \subseteq \QMLr$ because the polynomial reduction
would cause the length of the proof to increase polynomially.
In a recent article \citep{Beigi}, it has been shown how to obtain a soundness of $1-\frac{1}{n^{3+\epsilon}}$ for the language $\tSAT$ with two registers. 
\clearpage

The class $\QMLd$ is \emph{small} (included in $\QMA$) but still contains both $\NP$ and $\BPP$. This is an interesting
property since no relation is known between $\NP$ and $\BPP$. Showing that $\QMLd \subseteq \NP$ would somehow imply that classical non-determinism allows us to simulate a polynomial-size quantum circuit. 
 In this paper, we show that by slightly changing the class definition  the paradigm of
unentangled logarithmic size registers leads to a characterization of
the class $\NP$. We thus introduce a new complexity class $\BT$ and show that $\NP =\BT$. Once again, this class is defined to have a polynomially
small gap between completeness and soundness. Compared to our preliminary works, where we defined $\QMLd$, we do not consider
the verifier to work in quantum polynomial time but only to be able to generate a quantum circuit of polynomial size that acts on a logarithmic number of qubits.
This still allows the verifier to do the protocol such as in \citep{cancun} but also to define a classical polynomial size certificate for the class
$\BT$ wich imply $\BT \subseteq \NP$

%%%%%%%%%%%%%%%%%%%%%%%%%%%%%%%%%%%%%%%%%%%%%%%%%%%%%%%%%%%%%%%%%%%%%%%%%%%%%%%%%%%%%
\section{Definitions and Theorem}

A formal definition of the class $\BT$ (Classical polynomial-time quantum Merlin Arthur with two unentangled logarithmic size certificates) will
follow.  
Informally, it can be seen as the class of languages for which there exists a logarithmic quantum
proof with the promise that it is separated into two unentangled parts. 
The verifier works in classical polynomial time. It is allowed to produce
quantum circuits of polynomial size acting on a logarithmic number of qubits.

The following definition is simply a formal statement of what is usually referred to as a set of gate one can \emph{efficiently approximate}.

\begin{definition}
A {\em natural gate set} $\mathcal{U}$ is a finite set of unitary transformations acting on a finite number of qubits such that for all $U \in
\mathcal{U}$ there exists a classical algorithm that can approximate every element of the matrix $U$ up to $n$ bits in time polynomial in $n$.  
Furthermore $C(\mathcal{U})$ is the set of circuit composed of gates from $\mathcal{U}$. 
\end{definition}

With this definition in hand, we can define formally $\BT$. 

\clearpage

\begin{definition} \label{defclasse}
A language $L$ is in $\BT$ if there exists 
a natural gate set $\cal U$, polynomials $p$ and $q$,
a constant $c$ and a classical algorithm $\mathcal{V}$ running in
polynomial time that is allowed, for a word $x$ where $|x|=n$, to produce a quantum circuit $\mathcal{Q} = \mathcal{V}(x) \in C(\mathcal{U})$ of polynomial size $q(n)$
acting on $O(\log(n))$ qubits such that :\\

%1) (Completeness) if $x \in L$, there exists a state $\ket{w}=\ket{w_1} \otimes \ket{w_2} \in \left( {\cal H}_2^{\otimes c \log(n)}
%\right)^{\otimes2}$ s.t. 
%$\Pr[\mathcal{Q}(\ket{w})=accept]=1$; \\

1) (Completeness) if $x \in L$, there exists a state $\ket{w} \in \left( {\cal H}_2^{\otimes c \log(n)}
\right)^{\otimes2}$ s.t.
\begin{flushleft}

$\hspace{12mm} \Pr[\mathcal{Q}(\ket{w})=accept]=1$,\ where $\ket{w}=\ket{w_1} \otimes \ket{w_2}$; \\

\end{flushleft}

%2) (Soundness) if $x \not\in L$, with $|x|=n$, then for all states $\ket{w}=
%\ket{w_1} \otimes \ket{w_2}  \in \left( {\cal H}_2^{\otimes c \log(n)}
%\right)^{\otimes 2} $, 
%$\Pr[\mathcal{Q}(\ket{w})=accept]<1-\frac{1}{p(n)}$.

2) (Soundness) if $x \not\in L$, with $|x|=n$, then for all states $\ket{w} \in \left( {\cal H}_2^{\otimes c \log(n)}
\right)^{\otimes 2} $, 
\begin{flushleft}

$\hspace{12mm} \Pr[\mathcal{Q}(\ket{w})=accept]<1-\frac{1}{p(n)}$,\ where $\ket{w}=\ket{w_1} \otimes \ket{w_2}$.

\end{flushleft}
\end{definition}

The main result is that $\NP = \BT$.  This will be proven using the following well-known $\NP$-complete language:

\begin{definition} $\tCOL$ is the set of graphs $G=(V,E)$ (using any natural encoding into strings) for which there exists a coloring
$C:V\rightarrow\{0,1,2\}$ such that for all $(x,y)$ in $E$, $C(x) \neq C(y)$. 
\end{definition}

\begin{theorem} $\NP = \BT$
\end{theorem}

This will be proven in the two following sections. In the next section, we will describe an algorithm showing that $\tCOL$ is in $\BT$. It follows that
$\NP \subseteq \BT$. Then it will be proved that $\BT \subseteq \NP$.

%%%%%%%%%%%%%%%%%%%%%%%%%%%%%%%%%%%%%%%%%%%%%%%%%%%%%%%%%%%%%%%%%%%%%%%%%%%%%%%%%%%%%
\section{Logarithmic size quantum proof for $\tCOL$}

In this section, we will prove the following statement:

\begin{lemma}\label{includeinNP}
$\NP \subseteq \BT$
\end{lemma}

To prove the statement we will show that the language $\tCOL$ is in class $\BT$. 
We will address the completeness in \ref{completeness-thm} and the soundness in \ref{soundness-thm}. The proof is conclusive since $\tCOL$ is $\NP$-complete over polynomial time reduction. On the one hand, 
the soundness in $\BT$ only has to be polynomially close to one. On the other hand, the proof is of logarithmic size and the number of registers is
constant. Therefore, any decision problem in $\NP$ reduced to $\tCOL$ will still have a
protocol with a logarithmic-size proof and a satisfying soundness.

\subsection{Protocol and completeness}

We describe the verifier for the language $\tCOL$. The registers of the proof $\ket{\Psi}$ and $\ket{\Phi}$ are both regarded as vectors in ${\cal H}_n \otimes{\cal H}_3$, respectively the node and
color part of the register. The verifier performs one of the following three tests with equal probability. 
If the test succeeds, he accepts, otherwise he rejects.

\begin{itemize}
\item{{\bf Test 1:}} (Equality of the two registers) Perform the {\bf swap-test} \citep{finger} on $\ket{\Psi}$ and $\ket{\Phi}$ and reject if the
test fails.
\item{{\bf Test 2:}} (Consistency with the graph) $\ket{\Psi}$ and $\ket{\Phi}$ are measured in the computational basis, yielding $(i
,C(i)),(i',C'(i))$,  \\
a) if $i=i'$, verify that $C(i)=C'(i)$.  \\
b) otherwise if $(i,i')\in E$ verify that $C(i) \neq C'(i)$. 

\item{{\bf Test 3:}} (All nodes are present) For both $\ket{\Psi}$ and $\ket{\Phi}$, measure the index part of the register and the color part
separately in the Fourier basis. If the outcome of the measurement of the color part is $F_3 \ket{0}$ and the outcome of the index part is not $F_n \ket{0}$, then reject.
\end{itemize}

The following Lemma states that the protocol has completeness 1. 

\begin{lemma}
\label{completeness-thm} If $x \in$ $\tCOL$ then there exists a proof that the verifier described above will accept with probability 1.
\end{lemma}

\begin{proof} 
Let the quantum proof be $\ket{\Psi}=\ket{\Phi}= \frac{1}{\sqrt{n}}\sum_{i}\ket{i}\ket{C(i)}$ 
where $C$ is a valid coloring of the graph $G$. The probability that the {\bf swap-test} outputs {\em equal} is
$\demi+\frac{|\braket{\Psi}{\Phi}|^{2}}{2}$ \citep{finger}. Since $\ket{\Psi}=\ket{\Phi}$, the probability that Test 1 succeeds is 1.  Because $C$ is a
valid coloring of $G$, we have that Test 2 succeeds with probability 1. To see that Test 3 will also succeed with certainty, it is sufficient to see
that:
\begin{eqnarray*}
(I \tensor F_3) \frac{1}{\sqrt{n}}\sum_{j}\ket{j}\ket{c_{j}} = \frac{1}{\sqrt{n}}\sum_{j}\ket{j}\frac{1}{\sqrt{3}}\sum_{k}e^{2 \pi i c_{j} k }\ket{k}
\end{eqnarray*}
and therefore, if the color register is measured to be in the state $\ket{0}$, the resulting state will be $\frac{1}{\sqrt{n}}\sum_{i}\ket{i}=F_n
\ket{0}$.
\end{proof}
\clearpage

%%%%%%%%%%%%%%%%%%%%%%%%%%%%%%%%%%%%%%%%%%%%%%%%%%%%%%%%%%%%%%%%%%%%%%%%%%%%%%%%%%%%%
\subsection{Soundness}

Let us now consider the case where $G \not \in$ $\tCOL$. \ref{soundness-thm} at the end of this section states that if $G$ is not 3-colorable, then there is a non-negligible probability that one of
the three tests will fail. To prove this, we will require the following five simple Lemmas.

Because we know that the two registers given by the prover are not entangled, they can be written separately as $$
\ket{\Psi} = \sum_{i}\alpha_{i}\ket{i}\sum_{j}\beta_{i,j}\ket{j} \ \ \ \ \ \ \ 
\ket{\Phi} = \sum_{i}\alpha'_{i}\ket{i}\sum_{j}\beta'_{i,j}\ket{j}$$
where $\sum_{i}\prob{\alpha_{i}} = 1$ and $\forall i$, $\sum_{j}\prob{\beta_{i,j}} = 1$ 
and likewise for $\ket{\Phi}$. 
It is not difficult to see that the use of unentangled mixed states would not help the prover.

The following Lemmas will give us some useful facts on the behavior of the state when measured in the computational basis. The next Lemma says that if
Test 1 succeeds with high enough probability, then the distribution of outcomes will be similar for the two states.

\begin{lemma}
\label{lemma1} Let $\ket{\Psi}$ and $\ket{\Phi}$ be as defined earlier. If there exists a $k$ and an $l$ such that $\Huge{|}  |\alpha_k \beta_{k,l}|
-|\alpha'_k \beta'_{k,l}| \Huge{|} \geq 1/n^3$ then Test 1 will fail with probability at least $\frac{1}{8n^6}$.
\end{lemma}

\begin{proof} Let $P_{i,j}=|\alpha_i \beta_{i,j}|^2$ and $Q_{i,j}=|\alpha'_i \beta'_{i,j}|^2$ be the probability distributions when $\ket{\Phi}$ and
$\ket{\Psi}$ are measured in the computational basis. We will use the fact that, for any von Neumann measurement, the distances defined below are such
that $D(\ket{\Psi}, \ket{\Phi}) \ge D(P,Q)$, where $P$ and $Q$ are the classical outcomes distributions of the measurement. Then,
\begin{eqnarray*}
\sqrt{1-|\braket{\Psi}{\Phi}|^{2}}  
&{\stackrel{\text{\tiny def}}{=}}&  D(\ket{\Psi}, \ket{\Phi}) \\
&\geq &   D(P, Q)   \\
&{\stackrel{\text{\tiny def}}{=}}&       \demi \sum_{ij} \left||\alpha_i \beta_{i,j}|^{2} - |\alpha'_i \beta'_{i,j}|^{2}\right|\\
&\geq &   \demi\left||\alpha_k \beta_{k,l}|^{2} - |\alpha'_k \beta'_{k,l}|^{2}\right|\\
&\geq &   \demi \cdot \frac{1}{n^{3}}\\
\end{eqnarray*}

This means that $|\braket{\Psi}{\Phi}|^{2} \leq 1-\frac{1}{4n^6}$ and that Test 1 will fail with probability at least $\frac{1}{8n^6}$.
\end{proof}

The next Lemma states that nodes with a high enough probability of being observed have a well-defined color.

\begin{lemma}
\label{lemma2} Given that the quantum proof would pass both Test 1 and part a) of Test 2 with probability of failure no larger than $\frac{1}{8n^6}$, it must
be that $\forall i$ for which $|\alpha_{i}| \ge \frac{1}{n^{2}} \text{, } \exists! j \text{ such that } |\beta_{i,j}|^{2} \ge \frac{99}{100}$.
\end{lemma}

\begin{proof} Suppose for the sake of contradiction that there exists an $i$ such that 
$|\alpha_{i}|^{2} \ge \frac{1}{n^{2}}$ for which two of the $\beta_{i,j}$ 
have a squareed norm larger than $1/200$. 
Hence, w.l.o.g we can assume that $|\beta_{i,0}|^2> 1/200$ and $|\beta_{i,1}|^2> 1/200$. Because of \ref{lemma1}, we have that 
$\prob{\alpha'_{i}}\prob{\beta'_{i,1}} 
\ge \prob{\alpha_{i}}\prob{\beta_{i,1}} - \frac{1}{n^{3}} 
\ge \frac{1}{200n^{2}} - \frac{1}{n^{3}}$. 
Therefore, the probability of obtaining $(i,0)$ when measuring $\ket{\Psi}$ and $(i, 1)$ when measuring  $\ket{\Phi}$ is at least $$
\pare{\frac{1}{200n^{2}}} 
\pare{\frac{1}{200n^{2}} - \frac{1}{n^{3}}} \ge \frac{1}{8n^{6}}$$ 
when $n$ is large enough.  This is in contradiction with the hypothesis.  Therefore, the norm squared of the amplitude for two of the three colors must be less than
$\frac{1}{200}$, concluding the proof. 
\end{proof}

The next three Lemmas tell us what Test 3 actually implies.

\begin{lemma}
\label{lemma3} Given that the quantum proof would pass both Test 1 and part a) of Test 2 with probability of failure less than $\frac{1}{8n^{6}}$, then the
probability of measuring $\ket{\overline 0}=F_3 \ket{0}$ in the Fourier basis on the color register is greater than $1/5$ when $n$ is large enough.
\end{lemma}

\begin{proof} Assume that the node register is measured. If the outcome is $i$, then the probability of obtaining $\ket{\overline 0}$ in the Fourier
basis on the color register is given by $$
\frac{1}{3}|\beta_{i,0}+ \beta_{i,1}+\beta_{i,2}|^2.$$
For all $i$ with probability larger than $1/n^2$ \ref{lemma2} applies, in which case we can assume w.l.o.g that $|\beta_{i,0}|^2>99/100$ and
$|\beta_{i,1}|^2+|\beta_{i,2}|^2 \leq 1/100$. Using the Cauchy-Schwarz inequality, we obtain
\begin{eqnarray*} 
\frac{1}{3}|\beta_{i,0}+ \beta_{i,1}+\beta_{i,2}|^2 
 & \geq & \frac{1}{3}|  \ |\beta_{i,0}| -| \beta_{i,1}+\beta_{i,2}| \ | ^2 \\
 & \geq & \frac{1}{3}|  \ |\beta_{i,0}| -\sqrt{2 (| \beta_{i,1}|^2 +|\beta_{i,2}|^2) }\ | ^2 \\
 & \geq & \frac{1}{4} 
\end{eqnarray*}
Now, note that only $n-1$ of the nodes can have a probability smaller than $1/n^2$, and therefore the probability of obtaining 0 on the color register
is at least $(1-(n-1) \frac{1}{n^2})\frac{1}{4} \geq \frac{1}{5} $ for large enough $n$.
\end{proof}

\begin{lemma}
\label{lemma4} Given a state $\ket{X} = \sum_{i} \gamma_{i}\ket{i}$ such that there exist an $l$ with $\prob{\gamma_{l}} < \frac{1}{2n}$, then the
probability of not getting $\ket{\overline 0}=F_n \ket{0}$ when we measure $\ket{X}$ in the Fourier basis is at least $\frac{1}{16n^{2}}$.
\end{lemma}

\begin{proof} Let $P$ and $Q$ be the probability distributions when measuring $\ket{X}$ and 
$F_n \ket{0}$ respectively in the computational basis. Using the same techniques as in \ref{lemma1} we get:
\begin{eqnarray*}
\sqrt{1-|\braket{X}{\overline 0}|^2} 
&{\stackrel{\text{\tiny def}}{=}}  &D(\ket{X}, \ket{\overline 0}) \\
&\ge &D(P,Q) \\
&{\stackrel{\text{\tiny def}}{=}} &\demi \sum_{i} \left| \prob{\gamma_{i}} - \frac{1}{n} \right| \\
&\ge &\demi \left| \prob{\gamma_{l}} - \frac{1}{n} \right| \\
& \ge &\frac{1}{4n}\\
\end{eqnarray*}
This implies that the probability of failing the test is greater than  $\frac{1}{16n^{2}}$.
\end{proof}

\begin{lemma}
\label{lemma5} Assuming Test 1, Test 3 and part a) of Test 2 would succeed with probability larger than 
$\frac{1}{8n^{6}}$, then it must be that for all $i$, $\prob{\alpha_{i}} \geq \frac{1}{10n}$.
\end{lemma}

\begin{proof} Because of \ref{lemma3}, the probability of measuring 0 on the color register while performing Test 3 is at least $1/5$. 
Let $\ket{X} = \sum_{i} \gamma_{i}\ket{i}$ be the state after measuring 0 on the color register of $\ket{\Psi}$. Suppose that there was an $i$ in the original $\ket{\Psi}$ such that $|\alpha_i|^2<
1/(10n)$. 
Again, because of \ref{lemma3} we must have $|\gamma_i|^2<1/(2n)$. 
Now, from this fact and \ref{lemma4}, we conclude that $\ket{\Psi}$ would fail Test 3 with too large a probability. 
Therefore, for all $i$, $|\alpha_i|^2 \geq 1/(10n)$ .
\end{proof}

Now, using \ref{lemma1}, \ref{lemma2}, \ref{lemma3} and \ref{lemma5},  
it will be possible to prove the soundness of our verifier.

\begin{lemma}
\label{soundness-thm} If $x \not\in$ $\tCOL$ then all quantum proofs will fail the test with probability at least $\frac{1}{24n^6}$
\end{lemma}
\clearpage

\begin{proof} Assume that the graph $G$ is not 3 colorable and that it would fail Test 1, Test 3 and part a) of Test 2 with probability smaller than
$\frac{1}{8n^6}$. 
Let $C(i)={\max}_j |\beta_{ij}|$ be a coloring. 
Because of \ref{lemma1}  and \ref{lemma2}, this maximum is well defined. 
Since the graph is not 3-colorable, there exists two adjacent vertices $v_1$ and $v_2$ in $G$ such that $C(v_1)=C(v_2)$. Because of \ref{lemma5},
when performing Test 2 we have a probability of at least $1/(10n^2)$ of measuring $C(v_1)$ in the first register, and because of \ref{lemma1}
and \ref{lemma5}, we have a probability $1/(10n^2)-1/n^3$ of measuring  $C(v_2)$ in the second register. Combining these results, we have a probability
larger than $\frac{1}{8n^6}$ of failing condition b) of Test 2.
\end{proof}

%%%%%%%%%%%%%%%%%%%%%%%%%%%%%%%%%%%%%%%%%%%%%%%%%%%%%%%%%%%%%%%%%%%%%%%%%%%%%%%%%%%%%
\section{Polynomial size classical proof for languages in $\BT$}

In this section, we will prove the following statement.

\begin{lemma}\label{containsNP}
$\BT \subseteq \NP$
\end{lemma}
\begin{proof}

We will show that for all languages $L$ in $\BT$ there is a totally classical polynomial-time verifier $V'$ 
for the language.
The verifier $V'$ will receive a classical description (density matrix) of the quantum proof and 
compute the acceptance probability of the quantum circuit $\mathcal{Q}=V(x)$ with enough precision.

Let $g$ be the gap between the soundness and completeness for the verifier 
associated with the language $L$. Following the notation introduced in \ref{defclasse}, $g=1/p(n)$.

More precisely, the (totally) classical verifier $V'$ will do  the following in order to accept the language $L$.
Let us call the classical witness $\rho$.  
The verifier $V'$ first simulates $V$ (the verifier as defined for the class $\BT$) to compute a description of the quantum circuit 
$\mathcal{Q} = \mathcal{V}(x) \in C(\mathcal{U})$. 
Let $U$ be the unitary transformation corresponding to $\mathcal{Q}$ and 
$\Pi_{accept}$ be the projector corresponding to mesuring 1 on the accepting qubit.
The verifier $V'$ approximates the value $Tr(\Pi_{accept} U (\rho \otimes \rho)U^\dagger)$
with precision $g/3$.
The verifier $V'$ accepts if the acceptance probability is larger than $1-g/2$. 
This value can be evaluated in polynomial time for the number of qubits $\mathcal{Q}$ act on 
is logarithmic and the number of gates is polynomial. 
\end{proof}

%%%%%%%%%%%%%%%%%%%%%%%%%%%%%%%%%%%%%%%%%%%%%%%%%%%%%%%%%%%%%%%%%%%%%%%%%%%%%%%%%%%%%

\section{Conclusion}

In this paper, we proved that $\NP = \BT$. 
This gives a quantum characterization of one of the most important complexity classes. 
Moreover, this characterization is interesting in the sense that it is done within the paradigm of interactive proofs. 
This result gives us insight on the power of quantum information and the subtlety of entanglement.\\
For future work, on the one hand, a very natural way to improve on our result would be to investigate if  $\NP  = \BT$
even when the gap
between soundness and completeness is constant rather than only non-negligible.  
As stated before, it would be surprising since it could lead to the result $\QMAd = \NEXP$.  
As an intermediate result, one could try to show a better gap between completeness and soundness using a constant
number of unentangled parts in the proof. 
This is related to the question of whether $\QMLk = \QMLd$, and also reminiscent of the $\QMAk$ vs $\QMAd$ problem.

On the other hand, Aaronson et al. \citep{Aaronson} achieved a constant gap at the cost of requiring $\sqrt{n}\polylog(n)$ registers.  
Is it possible to use fewer unentangled registers and still achieve a constant gap? 
The ultimate goal would indeed be to show that all languages in $\NP$ have
logarithmic size quantum proofs with constant completeness and
soundness. If this it is not the case, what is the minimal length of such a proof in order to obtain a constant gap for
an $\NP$-complete problem?

\begin{acknowledge}
We would like to thank Tsuyoshi Ito for insightful discussion. 
We also would like to thank Fr\'ed\'eric Dupuis and Richard MacKenzie for their help with the language of
Shakespeare and the QuantumWorks project for funding this research.
\end{acknowledge}

\bibliography{texrefs}% Insert your bibliography file names here.

\end{document}